\documentclass{article}
\usepackage{graphicx}
\usepackage{amssymb}
\usepackage{amsmath}
\usepackage{amsthm}
\usepackage{rotating}
\usepackage[a4paper]{geometry}
\usepackage[utf8]{inputenc}
\usepackage[T1]{fontenc}
\usepackage{lmodern}
\allowdisplaybreaks[1] 
\usepackage{mathrsfs}
\usepackage{microtype}
\usepackage{latexsym}
\usepackage[colorlinks]{hyperref} 
\usepackage{cleveref} 
\usepackage{graphicx}
\usepackage{bookmark}
\usepackage{booktabs}
\usepackage{enumitem}
\usepackage{diagbox}
\usepackage{lineno} 
\usepackage{float}
\usepackage{color}
\usepackage{cite}
\usepackage{bm}
\usepackage{longtable}
\usepackage{xurl}

\newtheorem{theorem}{Theorem}[section]
\newtheorem{lemma}[theorem]{Lemma}
\newtheorem{proposition}[theorem]{Proposition}
\newtheorem{corollary}[theorem]{Corollary}

\theoremstyle{definition}
\newtheorem{definition}{Definition}[section]
\newtheorem{example}[definition]{Example}

\theoremstyle{remark}
\newtheorem{remark}{Remark}

\title{Shortest LCD embeddings of binary, ternary and quaternary linear codes}
\author{Junmin An\thanks{junmin0518@sogang.ac.kr, Department of Mathematics and Institute for Mathematical and Data Sciences, Sogang University, Seoul, Korea}, Ji-Hoon Hong\thanks{rjekfl@sogang.ac.kr, Department of Mathematics, Sogang University, Seoul, Korea}, Jon-Lark Kim\thanks{jlkim@sogang.ac.kr, Department of Mathematics and Institute for Mathematical and Data Sciences, Sogang University, Seoul, Korea}, Haeun Lim\thanks{haeunlim@sogang.ac.kr, Department of Mathematics and Institute for Mathematical and Data Sciences, Sogang University, Seoul, Korea}}
\date{}

\begin{document}
\maketitle

\begin{abstract}
In recent years, there has been active research on self-orthogonal embeddings of linear codes since they have yielded some optimal self-orthogonal codes. LCD codes have a trivial hull so they are counterparts of self-orthogonal codes. It is therefore a natural question whether one can embed linear codes into distance-optimal LCD codes. To answer it, we first determine the number of columns to be added to a generator matrix of a linear code in order to embed the given code into an LCD code. Then we characterize all possible forms of shortest LCD embeddings of a linear code. Using the shortest LCD embedding method, we find new ternary LCD codes with parameters including $[23, 4, 14]$, $[23, 5, 12]$, $[24, 6, 12]$, and $[25, 5, 14]$ and a new quaternary LCD $[21, 10, 8]$ code, each of which has minimum distance one greater than those of the known codes. This shows that our shortest LCD embedding method is useful in finding distance-optimal LCD codes over various fields.
\end{abstract}

\section{Introduction}

In 1992, J. L. Massey~\cite{Massey-LCD-code} defined LCD codes as linear codes with complementary duals. Since he showed that, for LCD codes, the nearest-codeword decoding problem could be reduced to a simpler problem, LCD codes have been the subject of wide interest. Sendrier~\cite{Sendrier-LCD-GV} showed that LCD codes met the asymptotic Gilbert-Varshamov bound, which indicates that they have good properties from the perspective of coding theory. Liu and Wang~\cite{Liu-LCD-class} enumerated the number of Euclidean and Hermitian LCD codes over finite fields using cogredience theories of alternative matrices, symmetric matrices. One of the most notable recent studies on LCD codes is their application to cryptography. Carlet and Guilley~\cite{Carlet-LCD-side-channel} introduced an application of LCD codes against side-channel attacks (SCA) and fault injection attacks (FIA).

Carlet et al.~\cite{Carlet-LCD-equiv} have shown an interesting result that every linear code over $\mathbb{F}_q$ where $q\ge 4$ is equivalent to some Euclidean LCD code and every linear code over $\mathbb{F}_{q^2}$ where $q\ge 3$ is equivalent to some Hermitian LCD code. However, in the case of Euclidean LCD codes over $\mathbb{F}_2$, $\mathbb{F}_3$ and Hermitian LCD codes over $\mathbb{F}_4$, the bounds for minimum distances are unknown for many lengths and dimensions.

Just as determining optimal minimum distances is central in coding theory, the same question is also fundamental to LCD codes. There has been extensive research in determining the largest minimum distances of LCD codes. Galvez et al.~\cite{Galvez-bound-LCD} determined the largest minimum distances of binary LCD codes for $n\le 12$. For binary LCD codes with $n\le 50$, partial results have been established in~\cite{Araya-optimal-LCD-2020,Bouyuklieva-optimal-LCD-2021,Harada-optimal-LCD-2019,Ishizuka-optimal-LCD-2023,Li-optimal-LCD-2024-1,Wang-optimal-LCD-2024-2}. For ternary LCD codes and quaternary Hermitian LCD codes, some largest minimum distances have been determined for lengths $n\le 25$ in~\cite{Araya-ternary-LCD-2020,Araya-ternary-LCD-2021,Lu-quaternary-LCD,Li-optimal-LCD-2024-1}. However, many cases still remain open, providing the motivation for us to study this problem.

A commonly used approach to determine the largest minimum distance of LCD codes is to explicitly construct codes meeting the best known bounds. Various construction methods have been proposed in~\cite{Araya-ternary-LCD-2021,Bouyuklieva-optimal-LCD-2021,Harada-optimal-LCD-2019,Ishizuka-optimal-LCD-2023,Kim-building-up,Lu-quaternary-LCD,Li-optimal-LCD-2024-1,Wang-optimal-LCD-2024-2}, and these methods have yielded new LCD codes achieving previously unknown largest minimum distances.

In this paper, we introduce an LCD embedding method for linear codes and use it to construct new LCD codes. Our approach is motivated by the following papers. Kim et al.~\cite{Kim-SO-embedding-2,Kim-SO-embedding-1} introduced the embedding approach to find new self-orthogonal codes. An et al.~\cite{An-SO-embedding} extended this method to arbitrary dimensions and obtained several new self-orthogonal codes.
Although an LCD embedding of a linear code is somewhat different from a self-orthogonal embedding of a linear code, we can determine the minimum number of columns (to be added) in order to embed a linear code into an LCD code, and characterize what kinds of columns or matrices should be added. Using our method, we have constructed four new ternary LCD codes with parameters including  $[23, 4, 14]$, $[23, 5, 12]$, $[24, 6, 12]$, and $[25, 5, 14]$, as well as a new quaternary LCD $[21, 10, 8]$ code.

The paper is organized as follows. In Section 2, we provide some preliminary concepts and notations in coding theory. In Section 3, we introduce the notion of LCD embedding, in particular, a shortest LCD embedding of a linear code. In Section 4, we construct new LCD codes by using our LCD embedding method applied to some distance-optimal linear codes. We conclude our paper in Section 5.

\section{Preliminaries}

In this section, we introduce basic concepts in coding theory that will be used in the subsequent developments. The books \cite{Huffman-Coding2} and \cite{Macwilliams-Coding1} provide general references on coding theory.

Let $\mathbb{F}_q$ be the finite field of order $q$ where $q$ is a prime power. A \textit{code} $\mathcal{C}$ of length $n$ over $\mathbb{F}_q$ is a subset of $\mathbb{F}_q^n$. If $\mathcal{C}$ is a subspace of $\mathbb{F}_q^n$, then we call $\mathcal{C}$ an $[n, k]$ \textit{linear code} where $k$ is the dimension of $\mathcal{C}$ over $\mathbb{F}_q$. Elements of a code $\mathcal{C}$ are called \textit{codewords}. A \textit{generator matrix} $G(\mathcal{C})$ of a code $\mathcal{C}$ is a matrix over $\mathbb{F}_q$ whose rows form a basis for $\mathcal{C}$. A monomial matrix is a matrix in which each row and each column contains exactly one nonzero entry. Two (linear) codes $\mathcal{C}_1$ and $\mathcal{C}_2$ over $\mathbb{F}_q$ are called \textit{equivalent} if there is a monomial matrix $M$ such that $\mathcal{C}_2=\mathcal{C}_1M$.

For two vectors $\mathbf{u}=(u_1, \ldots, u_n)$ and $\mathbf{v}=(v_1, \ldots, v_n)$ in $\mathbb{F}_q^n$, the \textit{Euclidean inner product} between them is defined as
\[
\langle \mathbf{u}, \mathbf{v}\rangle_E=\sum_{i=1}^nu_iv_i.
\]
Similarly, for vectors $\mathbf{u}$ and $\mathbf{v}$ in $\mathbb{F}_{q^2}^n$, the \textit{Hermitian inner product} between $\mathbf{u}$ and $\mathbf{v}$ is defined as
\[
\langle \mathbf{u}, \mathbf{v}\rangle_H=\sum_{i=1}^nu_iv_i^q.
\]
For a code $\mathcal{C}$ of length $n$ over $\mathbb{F}_q$, the \textit{Euclidean dual} of $\mathcal{C}$ is defined as
\[
\mathcal{C}^{\perp_E}=\{\mathbf{x}\in\mathbb{F}_q^n\mid\langle\mathbf{x}, \mathbf{c}\rangle_E=0\text{ for all }\mathbf{c}\in\mathcal{C}\},
\]
and for a code $\mathcal{C}$ over $\mathbb{F}_{q^2}$, the \textit{Hermitian dual} of $\mathcal{C}$ is defined as
\[
\mathcal{C}^{\perp_H}=\{\mathbf{x}\in\mathbb{F}_{q^2}^n\mid\langle\mathbf{x}, \mathbf{c}\rangle_H=0\text{ for all }\mathbf{c}\in\mathcal{C}\}.
\]
The notions we introduce next apply to both the Euclidean dual and the Hermitian dual. We use the notation $\mathcal{C}^\perp$ to denote either the Euclidean dual or the Hermitian dual of $\mathcal{C}$. A linear code $\mathcal{C}$ is called \textit{self-orthogonal} if $\mathcal{C}\subseteq \mathcal{C}^{\perp}$, and is called \textit{self-dual} if $\mathcal{C}=\mathcal{C}^\perp$. For a linear code $\mathcal{C}$, the \textit{hull} of $\mathcal{C}$ is defined as the intersection of $\mathcal{C}$ and $\mathcal{C}^\perp$, and is denoted by $\text{Hull}(\mathcal{C})$. If $\mathcal{C}$ is self-orthogonal, then $\mathcal{C}=\text{Hull}(\mathcal{C})$. We call $\mathcal{C}$ an \textit{LCD code} if $\text{Hull}(\mathcal{C})=\{\mathbf{0}\}$.

For a linear code $\mathcal{C}$ over $\mathbb{F}_q$, $\text{Hull}(\mathcal{C})$ is also a linear code over $\mathbb{F}_q$. The following theorem gives the dimension of $\text{Hull}(\mathcal{C})$.

\begin{theorem}[\cite{Li-hull-dim}] \label{hull-formula}
    Let $\mathcal{C}$ be an $[n, k]$ linear code over $\mathbb{F}_q$ with generator matrix $G$. Then the dimension $\ell$ of the Euclidean hull is given by
    \[
    \ell=k-\rm{rank}(GG^T),
    \]
    where $G^T$ is the transpose of $G$. Furthermore, for a linear $[n, k]$ code $\mathcal{C}$ over $\mathbb{F}_{q^2}$, the dimension $\ell_h$ of the Hermitian hull is given by
    \[
    \ell_h=k-\rm{rank}(G\sigma(G)^T),
    \]
    where $\sigma(G)$ is the matrix obtained by taking the $q$-th power of each entry of $G$.
\end{theorem}
It follows that if $\mathcal{C}$ is Euclidean (resp. Hermitian) self-orthogonal, then the rank of $GG^T$ (resp. $G\sigma(G)^T$) equals $0$. On the other hand, if $\mathcal{C}$ is a Euclidean (resp. Hermitian) LCD code, then $k=\text{rank}(GG^T)$ (resp. $k=\text{rank}(G\sigma(G)^T)$), which means that $GG^T$ (resp. $G\sigma(G)^T$) is an invertible matrix.

The \textit{Hamming weight} of a vector $\mathbf{u}$ in $\mathbb{F}_q^n$ is defined as the number of nonzero coordinates of $\mathbf{u}$. The \textit{Hamming distance} between two vectors $\mathbf{u}$ and $\mathbf{v}$ in $\mathbb{F}_q^n$ is defined as the Hamming weight of $\mathbf{u}-\mathbf{v}$. The \textit{minimum (Hamming) distance} of a code $\mathcal{C}$ over $\mathbb{F}_q$ is the minimum of the distances between any two distinct codewords in $\mathcal{C}$. The \textit{minimum (Hamming) weight} of a code $\mathcal{C}$ over $\mathbb{F}_q$ is the minimum of the weights of any nonzero codeword in $\mathcal{C}$. If $\mathcal{C}$ is linear, then the minimum distance of $\mathcal{C}$ and the minimum weight of $\mathcal{C}$ are equal.

A linear $[n, k]$ code $\mathcal{C}$ over $\mathbb{F}_q$ is called \textit{distance-optimal} if its minimum distance is the largest among all linear codes of the same length and dimension. Similarly, an $[n, k]$ LCD code is called a \textit{distance-optimal LCD code} if it achieves the largest minimum distance among all $[n, k]$ LCD codes over $\mathbb{F}_q$. We denote the minimum distance of a distance-optimal $[n, k]$ LCD code over $\mathbb{F}_q$ by $d_{\text{LCD}}(n, k)$.

In this paper, we mainly consider $\mathbb{F}_2$, $\mathbb{F}_3$ and $\mathbb{F}_4$. For codes over $\mathbb{F}_2$ and $\mathbb{F}_3$, the dual refers to the Euclidean dual, while for codes over $\mathbb{F}_4$, the dual refers to the Hermitian dual. Similarly, LCD codes over $\mathbb{F}_2$ and $\mathbb{F}_3$ refer to the Euclidean LCD codes, and LCD codes over $\mathbb{F}_4$ refer to the Hermitian LCD codes. With these conventions, equivalent codes have hulls of the same dimension.

\section{Shortest LCD embeddings of codes}

Throughout this section, let $F$ denote $\mathbb{F}_2$, $\mathbb{F}_3$ or $\mathbb{F}_4$. For a matrix $G$ over $F$, we define $G^*$ as follows: $G^*=G^T$ if $F$ is $\mathbb{F}_2$ or $\mathbb{F}_3$, and $G^*=\sigma(G)^T$ if $F$ is $\mathbb{F}_4$. Here, $\sigma(G)$ denotes the matrix obtained by squaring each entry of $G$.

\begin{definition}
    Let $\mathcal{C}$ be an $[n, k]$ code over $F$. Then an $[n', k]$ LCD code $\tilde{\mathcal{C}}$ over $F$ with $n'\ge n$ is called an \textit{LCD embedding} of $\mathcal{C}$ if there is a subset $S$ of coordinates of $\tilde{\mathcal{C}}$ such that $\mathcal{C}$ is obtained by puncturing $\tilde{\mathcal{C}}$ on the coordinates in $S$.
\end{definition}

It follows from \cite[Theorem 15]{Liu-LCD-class} and \cite[Proposition 2]{Massey-LCD-code} that every $[n,k]$ code over $F$ has an LCD embedding. Let $\tilde{\mathcal{C}}$ be an LCD embedding of an $[n, k]$ code $\mathcal{C}$ over $F$. We call $\tilde{\mathcal{C}}$ a \textit{shortest LCD embedding} of $\mathcal{C}$ if its length is minimal among all LCD embeddings of $\mathcal{C}$.

\begin{theorem}\label{shortestLCD}
    Let $\mathcal{C}$ be an $[n, k]$ code over $F$ with $\ell=\dim({\rm Hull}(\mathcal{C}))$. Then the length of a shortest LCD embedding of $\mathcal{C}$ is $n+\ell$.
\end{theorem}
\begin{proof}
Let $G$ be a generator matrix of $\mathcal{C}$. Suppose that the length of a shortest LCD embedding $\tilde{\mathcal{C}}$ of $\mathcal{C}$ is $n+m$. Then, there is a $k\times m$ matrix $D$ such that $\tilde{G}=[G\mid D]$ generates $\tilde{\mathcal{C}}$. Since $\tilde{\mathcal{C}}$ is an LCD code, we have
\begin{align*}
k=\text{rank}(\tilde{G}\tilde{G}^*)&=\text{rank}(GG^*+DD^*)\\&\le \text{rank}(GG^*)+\text{rank}(DD^*)\\&\le \text{rank}(GG^*)+\text{rank}(D) \\ 
& \le (k-\ell)+m {\mbox{~(by Theorem~\ref{hull-formula})}}.
\end{align*}
Hence we have $m\ge \ell$.

Conversely, let $G'$ be a generator matrix of $\mathcal{C}$ of the form
\[
G'=\begin{bmatrix}
    G(\text{Hull}(\mathcal{C}))\\A
\end{bmatrix}.
\]
Let us take
\[
D=\begin{bmatrix}
    I_\ell\\\mathcal{O}
\end{bmatrix},
\]
where $I_\ell$ is the $\ell\times \ell$ identity matrix. Then we have
\[
[G'\mid D][G'\mid D]^*=\begin{bmatrix}
    G(\text{Hull}(\mathcal{C})) & I_\ell\\A & \mathcal{O}
\end{bmatrix}\begin{bmatrix}
    G(\text{Hull}(\mathcal{C}))^* & A^*\\I_\ell & \mathcal{O}
\end{bmatrix}=\begin{bmatrix}
    I_\ell & \mathcal{O}\\\mathcal{O} & AA^*
\end{bmatrix},
\]
which is clearly invertible. It follows that $[G'\mid D]$ generates an LCD code. Therefore, we have $m=\ell$.
\end{proof}
As one can see in the proof, Theorem~\ref{shortestLCD} applies not only to codes over $\mathbb{F}_2$, $\mathbb{F}_3$, and $\mathbb{F}_4$, but also to codes over any finite field.

\begin{corollary}
    Let $\mathcal{C}$ be an $[n, k]$ code over $F$ with $\ell=\dim({\rm Hull}(\mathcal{C}))$. Then the length of a shortest LCD embedding of $\mathcal{C}$ is equal to the length of a shortest LCD embedding of ${\rm Hull}(\mathcal{C})$.
\end{corollary}
\begin{proof}
    By Theorem~\ref{shortestLCD}, the lengths of shortest LCD embeddings of $\mathcal{C}$ and $\text{Hull}(\mathcal{C})$ are both $n+\ell$.
\end{proof}

\begin{remark}\label{Hamming-743}
    Let $\mathcal{C}$ be an $[n, k]$ binary code with generator matrix
    \[
    G=\begin{bmatrix}
    G(\text{Hull}(\mathcal{C}))\\A
\end{bmatrix}.
    \]
    Not every shortest LCD embedding of $\mathcal{C}$ has a generator matrix of the form
    \begin{equation}\label{stform}
    \tilde{G}=\begin{bmatrix}
    G(\text{Hull}(\mathcal{C})) & D\\A & \mathcal{O}
    \end{bmatrix}
    \end{equation}
    for some matrix $D$ over $F$.
    Consider the binary $[7, 4, 3]$ Hamming code $\mathcal{H}_3$ with generator matrix
    \[
    G=\begin{bmatrix}
        1 & 0 & 1 & 0 & 1 & 0 & 1\\0 & 1 & 1 & 0 & 1 & 1 & 0\\0 & 0 & 0 & 1 & 1 & 1 & 1\\1 & 0 & 0 & 0 & 1 & 1 & 0
    \end{bmatrix}
    \]
    whose first 3 rows form a generator matrix of the $[7, 3, 4]$ simplex code. By adding three columns, we obtain the following generator matrix for a shortest LCD embedding:
    \[
    \tilde{G}=\begin{bmatrix}
        1 & 0 & 1 & 0 & 1 & 0 & 1 & 0 & 0 & 1\\
        0 & 1 & 1 & 0 & 1 & 1 & 0 & 1 & 1 & 1\\
        0 & 0 & 0 & 1 & 1 & 1 & 1 & 0 & 1 & 0\\
        1 & 0 & 0 & 0 & 1 & 1 & 0 & 1 & 0 & 1
    \end{bmatrix},
    \]
    which generates a $[10, 4, 4]$ LCD code. This code cannot be represented by a matrix of the form given in \eqref{stform}. This is because if we append the zero vector to the $4$th row of $G$, then it would generate an LCD embedding of $\mathcal{H}_3$ with minimum distance $3$.
\end{remark}

\begin{lemma}\label{D-inv}
    Let $\mathcal{C}$ be an $[n, k]$ code over $F$ with $\ell=\dim({\rm Hull}(\mathcal{C}))$ and $\tilde{\mathcal{C}}$ be a shortest LCD embedding of $\mathcal{C}$ with generator matrix
    \[
    \tilde{G}=\begin{bmatrix}
        G({\rm Hull}(\mathcal{C})) & D\\A & C
    \end{bmatrix}.
    \]
    Then $D$ is invertible.
\end{lemma}
\begin{proof}
    Suppose that $\text{rank}(D)< \ell$. Then, there is a nonzero vector $\mathbf{u}\in F^\ell$ such that $\mathbf{u}D=\mathbf{0}$. Let
    \[
    \mathbf{c}=\mathbf{u}\cdot[G(\text{Hull}(\mathcal{C}))~|~D]=(\mathbf{h}, \mathbf{0}),
    \]
    where $\mathbf{h}=\mathbf{u}G(\text{Hull}(\mathcal{C}))$. Since the nonzero codeword $\mathbf{c}$ is orthogonal to every row of $\tilde{G}$, which generates the LCD code $\tilde{\mathcal{C}}$, this is a contradiction. Thus $\text{rank}(D)=\ell$.
\end{proof}

\begin{lemma}[\cite{Zhang-Schur}]\label{Schur_lemma}
Let $A$, $B$, $C$, and $D$ be matrices of sizes $n\times n$, $n\times m$, $m\times n$, and $m\times m$, respectively, over a field $F$. If $\det(A)\ne 0$, then the determinant $\det(M)$ of the matrix
\[
M=\begin{bmatrix}
    A & B\\C & D
\end{bmatrix}
\]
is given by $\det (M)=\det (A)\cdot \det(D-CA^{-1}B)$.
\end{lemma}

\begin{theorem}\label{construct-LCD}
    Let $\mathcal{C}$ be an $[n, k]$ code over $F$ with $\ell=\dim({\rm Hull}(\mathcal{C}))$ and generator matrix
    \[
    G=\begin{bmatrix}
    G({\rm Hull}(\mathcal{C}))\\A
    \end{bmatrix}.
    \]
    Let $\tilde{G}$ be a matrix obtained by appending an invertible $\ell\times \ell$ matrix $D$ and a $(k-\ell)\times \ell$ matrix $C$ as follows:
    \[
    \tilde{G}=\begin{bmatrix}
        G({\rm Hull}(\mathcal{C})) & D\\A & C
    \end{bmatrix}.
    \]
    Then $\tilde{G}$ generates a shortest LCD embedding of $\mathcal{C}$. Conversely, every shortest LCD embedding of $\mathcal{C}$ has a generator matrix of this form.
\end{theorem}
\begin{proof}
    Note that
    \[
    \tilde{G}\tilde{G}^*=\begin{bmatrix}
        G(\text{Hull}(\mathcal{C})) & D\\A & C
    \end{bmatrix}\begin{bmatrix}
        G(\text{Hull}(\mathcal{C}))^* & A^*\\D^* & C^*
    \end{bmatrix}=\begin{bmatrix}
        DD^* & DC^*\\CD^* & AA^*+CC^*
    \end{bmatrix}.
    \]
    By Lemma~\ref{Schur_lemma},
    \begin{align*}
    \det(\tilde{G}\tilde{G}^*)&=\det(DD^*)\cdot\det((AA^*+CC^*)-CD^*(DD^*)^{-1}DC^*)\\&=\det(DD^*)\cdot\det((AA^*+CC^*)-CC^*)\\&=\det(DD^*)\cdot\det(AA^*)\ne 0.
    \end{align*}
    Thus, $\tilde{G}$ generates an LCD code. By Theorem~\ref{shortestLCD}, it is a shortest LCD embedding. It follows from Lemma~\ref{D-inv} that every shortest LCD embedding of $\mathcal{C}$ has a generator matrix of the form
    \[
    \tilde{G}=\begin{bmatrix}
        G(\text{Hull}(\mathcal{C})) & D\\A & C
    \end{bmatrix},
    \]
    where $D$ is invertible.
\end{proof}

\begin{proposition}
    Let $\mathcal{C}$ be an $[n, k]$ code over $F$ with $\ell=\dim({\rm Hull}(\mathcal{C}))$ and $\tilde{\mathcal{C}}$ be a shortest LCD embedding of $\mathcal{C}$. Then $\tilde{\mathcal{C}}$ has a generator matrix of the form
    \[
    \tilde{G}=\begin{bmatrix}
        G({\rm Hull}(\mathcal{C}))' & I_\ell\\A' & \mathcal{O}
    \end{bmatrix},
    \]
    where $G({\rm Hull}(\mathcal{C}))'$ generates ${\rm Hull}(\mathcal{C})$ and
    \[
    G=\begin{bmatrix}
        G({\rm Hull}(\mathcal{C}))'\\A'
    \end{bmatrix}
    \]
    is a generator matrix of $\mathcal{C}$.
\end{proposition}
\begin{proof}
    By Theorem~\ref{construct-LCD}, $\tilde{\mathcal{C}}$ has a generator matrix of the form
    \[
    \tilde{G}_0=\begin{bmatrix}
        G(\text{Hull}(\mathcal{C})) & D\\A & C
    \end{bmatrix},
    \]
    where $G(\text{Hull}(\mathcal{C}))$ is a generator matrix of $\text{Hull}(\mathcal{C})$ and $D$ is an invertible matrix. Define
     \[
    M=\begin{bmatrix}
        I_\ell & \mathcal{O}\\-C & I_{k-\ell}
    \end{bmatrix}\quad\mbox{and}\quad N=\begin{bmatrix}
        D^{-1} & \mathcal{O}\\\mathcal{O} & I_{k-\ell}
    \end{bmatrix}.
    \]
    Since both $M$ and $N$ are invertible,
    \[
    \tilde{G}=MN\tilde{G}_0=\begin{bmatrix}
        D^{-1}G(\text{Hull}(\mathcal{C})) & I_\ell\\A-CD^{-1}G(\text{Hull}(\mathcal{C})) & \mathcal{O}
    \end{bmatrix}
    \]
    generates $\tilde{\mathcal{C}}$. It follows that the matrix
    \[
    G=\begin{bmatrix}
        D^{-1}G(\text{Hull}(\mathcal{C}))\\A-CD^{-1}G(\text{Hull}(\mathcal{C}))
    \end{bmatrix}
    \]
    which is a punctured matrix of $\tilde{G}$ on the last $\ell$ columns generates $\mathcal{C}$ such that $D^{-1}G(\text{Hull}(\mathcal{C}))$ generates $\text{Hull}(\mathcal{C})$.
\end{proof}

\begin{lemma}[\cite{Araya-quaternary-LCD,Niederreiter-ternary-simplex}]\label{simplex-orthogonal}
Let $\mathcal{S}_{q, r}$ be the $q$-ary simplex $[(q^r-1)/(q-1),~ r,~ q^{r-1}]$ code. Then $\mathcal{S}_{3, r}$ is Euclidean self-orthogonal for $r\ge 2$ and $\mathcal{S}_{4, r}$ is Hermitian self-orthogonal for $r\ge 2$.
\end{lemma}

\begin{example}
Let $\mathcal{H}_{q, r}$ be the $q$-ary $[(q^r-1)/(q-1), (q^r-1)/(q-1)-r, 3]$ Hamming code with $q\in\{2, 3, 4\}$. Let $r\ge 3$ if $q=2$, and $r\ge 2$ if $q=3$ or $q=4$. Since the dual of $\mathcal{H}_{q, r}$ is the $q$-ary simplex code $\mathcal{S}_{q, r}$ which is self-orthogonal by Lemma~\ref{simplex-orthogonal}, we have
\[
\dim(\text{Hull}(\mathcal{H}_{q, r}))=\dim\mathcal{S}_{q, r}=r.
\]
By Theorem~\ref{shortestLCD}, a shortest LCD embedding of $\mathcal{H}_{q, r}$ is a $[(q^r-1)/(q-1)+r, (q^r-1)/(q-1)-r]$ code.

As a shortest LCD embedding of $\mathcal{H}_{2, 4}$, we have obtained a $[19, 11, 4]$ LCD code whose generator matrix is given as follows:
    \[
    \setlength{\arraycolsep}{1.2pt}
    G=\left[
    \begin{array}{ccccccccccccccc|cccc}
    1&0 & 1 & 0 & 1 & 0 & 1 & 0 & 1 & 0 & 1 & 0 & 1 & 0 & 1 ~&~ 1 & 0 & 0 & 1 \\
    0 & 1 & 1 & 0 & 0 & 1 & 1 & 0 & 0 & 1 & 1 & 0 & 0 & 1 & 1 ~&~ 0 & 0 & 1 & 1 \\
    0 & 0 & 0 & 1 & 1 & 1 & 1 & 0 & 0 & 0 & 0 & 1 & 1 & 1 & 1 ~&~ 1 & 1 & 0 & 0 \\
    0 & 0 & 0 & 0 & 0 & 0 & 0 & 1 & 1 & 1 & 1 & 1 & 1 & 1 & 1 ~&~ 0 & 0 & 0 & 1 \\
    \hline
    1 & 0 & 0 & 0 & 0 & 0 & 0 & 0 & 0 & 0 & 0 & 0 & 0 & 1 & 1 ~&~ 0 & 0 & 0 & 1 \\
    1 & 1 & 0 & 0 & 0 & 0 & 0 & 0 & 0 & 0 & 0 & 0 & 1 & 1 & 0 ~&~ 0 & 0 & 0 & 0 \\
    0 & 1 & 1 & 0 & 0 & 0 & 0 & 0 & 0 & 0 & 0 & 0 & 0 & 1 & 1 ~&~ 0 & 0 & 0 & 0 \\
    0 & 0 & 1 & 1 & 0 & 0 & 0 & 0 & 0 & 0 & 1 & 0 & 1 & 1 & 1 ~&~ 0 & 0 & 0 & 0 \\
    0 & 0 & 0 & 1 & 1 & 0 & 0 & 0 & 0 & 0 & 0 & 0 & 0 & 1 & 1 ~&~ 1 & 0 & 0 & 1 \\
    0 & 0 & 0 & 0 & 1 & 1 & 0 & 0 & 0 & 0 & 0 & 0 & 1 & 1 & 0 ~&~ 1 & 0 & 0 & 1 \\
    0 & 0 & 0 & 0 & 0 & 1 & 1 & 0 & 0 & 0 & 0 & 0 & 0 & 1 & 1 ~&~ 0 & 0 & 0 & 0
    \end{array}
    \right].
    \]
These $[19, 11, 4]$ codes are distance-optimal and inequivalent to the code in the BKLC database.
\end{example}

\begin{example}
    Let $\mathcal{R}_q(r, m)$ be the $q$-ary $r$-th order generalized Reed-Muller $[n, k]$ code where $n=q^m$ and $q\in\{2, 3, 4\}$. If $r\le\lfloor\frac{m(q-1)-1}{2}\rfloor$, then $\mathcal{R}_q(r, m)$ is self-orthogonal. Thus the length of a shortest LCD embedding of $\mathcal{R}_q(r, m)$ is $n+k$. On the other hand, if $r>\lfloor\frac{m(q-1)-1}{2}\rfloor$, then $\mathcal{R}_q(r, m)$ contains its dual. Since the hull of $\mathcal{R}_q(r, m)$ is its dual, it follows that the length of a shortest LCD embedding of $\mathcal{R}_q(r, m)$ is $n+(n-k)=2n-k$.

    Since $1\le \lfloor\frac{3(3-1)-1}{2}\rfloor=2$, a ternary first-order Reed-Muller code $\mathcal{R}_3(1, 3)$ is a self-orthogonal code with parameters $[27, 4, 18]$. Thus, by Theorem~\ref{shortestLCD}, the length of a shortest LCD embedding of $\mathcal{R}_3(1, 3)$ is $31$. By computational search, we have verified that there are four inequivalent shortest LCD embeddings of $\mathcal{R}_3(1, 3)$, and all of them have minimum distance $19$. A generator matrix of one of these codes is given as follows:
    \[
    \setlength{\arraycolsep}{1.2pt}
    G=\left[
    \begin{array}{ccccccccccccccccccccccccccc|cccc}
    1 & 1 & 1 & 1 & 1 & 1 & 1 & 1 & 1 & 1 & 1 & 1 & 1 & 1 & 1 & 1 & 1 & 1 & 1 & 1 & 1 & 1 & 1 & 1 & 1 & 1 & 1 ~&~ 1 & 1 & 1 & 0 \\
    0 & 0 & 0 & 0 & 0 & 0 & 0 & 0 & 0 & 1 & 1 & 1 & 1 & 1 & 1 & 1 & 1 & 1 & 2 & 2 & 2 & 2 & 2 & 2 & 2 & 2 & 2 ~&~ 0 & 1 & 0 & 0 \\
    0 & 0 & 0 & 1 & 1 & 1 & 2 & 2 & 2 & 0 & 0 & 0 & 1 & 1 & 1 & 2 & 2 & 2 & 0 & 0 & 0 & 1 & 1 & 1 & 2 & 2 & 2 ~&~ 0 & 0 & 1 & 0 \\
    0 & 1 & 2 & 0 & 1 & 2 & 0 & 1 & 2 & 0 & 1 & 2 & 0 & 1 & 2 & 0 & 1 & 2 & 0 & 1 & 2 & 0 & 1 & 2 & 0 & 1 & 2 ~&~ 0 & 0 & 0 & 1
    \end{array}
    \right].
    \]
\end{example}

\section{Optimal LCD codes}

\begin{proposition}\label{ternary-new-code}
    There exists a ternary distance-optimal LCD $[23, 4, 14]$ code. There also exist ternary LCD codes with parameters including $[23, 5, 12]$, $[24, 6, 12]$ and $[25, 5, 14]$ whose minimum distances are one greater than those of the known LCD codes.
\end{proposition}
\begin{proof}
From the BKLC (best known linear codes) library of MAGMA, we obtain ternary $[19, 4, 12]$, $[19, 5, 11]$, $[20, 6, 10]$ and $[20, 5, 12]$ codes. Using Theorem~\ref{construct-LCD}, we construct the following codes, respectively.
\begin{itemize}
    \item Ternary $[23, 4, 14]$ LCD code $\mathcal C_{3,1}$ with generator matrix
    \[
    \setlength{\arraycolsep}{1.2pt}
    \left[
    \begin{array}{ccccccccccccccccccc|cccc}
    1&0&0&0&2&1&1&2&0&1&0&0&1&2&1&0&1&1&1~&~1&1&2&1\\
    0&1&0&0&2&2&2&1&2&0&0&2&1&1&1&1&1&0&0~&~2&0&1&0\\
    0&0&1&0&2&1&0&1&2&1&1&2&0&2&2&2&0&1&0~&~0&2&1&2\\
    0&0&0&1&0&2&0&2&1&1&1&2&2&2&0&1&1&0&1~&~2&1&0&2
    \end{array}
    \right].
    \]
    For ternary LCD $[23, 4]$ codes, the possible minimum distances were 13-14~\cite{Li-optimal-LCD-2024-1}.
    \smallskip
    \item Ternary $[23, 5, 12]$ LCD code $\mathcal C_{3,2}$ with generator matrix
    \[
    \setlength{\arraycolsep}{1.2pt}
    \left[
    \begin{array}{ccccccccccccccccccc|cccc}
    1&0&0&0&2&1&1&2&0&1&0&0&1&2&1&0&1&1&1~&~0&0&2&2\\
    0&1&0&0&2&2&2&1&2&0&0&2&1&1&1&1&1&0&0~&~2&0&1&1\\
    0&0&1&0&2&1&0&1&2&1&1&2&0&2&2&2&0&1&0~&~0&2&0&2\\
    0&0&0&1&0&2&0&2&1&1&1&2&2&2&0&1&1&0&1~&~2&1&0&2\\
    \hline
    0&0&0&0&1&1&1&1&0&2&1&1&2&0&0&1&0&1&1&~1&2&2&2
    \end{array}
    \right].
    \]
       For ternary LCD $[23, 5]$ codes, the possible minimum distances were 11-13~\cite{Li-optimal-LCD-2024-1}.
         \smallskip
    \item Ternary $[24, 6, 12]$ LCD code $\mathcal C_{3,3}$ with generator matrix
    \[
    \setlength{\arraycolsep}{1.2pt}
    \left[
    \begin{array}{cccccccccccccccccccc|cccc}
    1&0&0&0&0&2&2&1&1&0&0&2&0&2&1&1&2&2&2&0~&~0&0&0&1\\
    0&1&0&2&0&1&0&0&2&0&2&1&1&0&2&1&2&1&0&2~&~2&2&2&1\\
    0&0&1&2&0&0&1&2&2&2&0&2&1&0&2&0&2&0&1&1~&~1&0&2&0\\
    0&0&0&0&1&1&1&1&2&2&2&2&1&1&1&1&0&0&0&0~&~0&1&1&1\\
    \hline
    0&0&0&1&0&2&0&1&1&1&0&0&2&1&0&2&0&0&1&1~&~0&1&0&2\\
    0&0&0&0&0&0&1&2&1&1&0&1&1&1&0&1&2&2&0&2~&~1&1&0&2
    \end{array}
    \right].
    \]
        For ternary LCD $[24, 6]$ codes, the possible minimum distances were 11-13~\cite{Li-optimal-LCD-2024-1}.
      \smallskip
    \item Ternary $[25, 5, 14]$ LCD code $\mathcal C_{3,4}$ with generator matrix
    \[
    \setlength{\arraycolsep}{1.2pt}
    \left[
    \begin{array}{cccccccccccccccccccc|ccccc}
    1&0&0&0&0&2&2&0&0&0&1&1&0&2&1&1&1&2&2&1~&~2&1&0&2&0\\
    0&1&0&0&0&0&0&2&2&2&1&0&0&1&1&2&1&1&1&1~&~1&2&1&2&1\\
    0&0&1&0&0&2&1&2&2&0&2&0&2&2&2&0&0&2&1&1~&~2&0&2&0&0\\
    0&0&0&1&0&2&0&2&1&1&1&2&2&2&0&1&1&0&1&0~&~0&2&1&0&0\\
    0&0&0&0&1&1&1&1&0&2&1&1&2&0&0&1&0&1&1&1~&~2&1&2&2&2
    \end{array}
    \right].
    \]
          For ternary LCD $[25, 5]$ codes, the possible minimum distances were 13-15~\cite{Li-optimal-LCD-2024-1}.
\end{itemize}
This completes the proof.
\end{proof}

\begin{proposition}\label{quaternary-new-code}
    There exists a quaternary distance-optimal LCD $[21, 10, 8]$ code, denoted by $\mathcal C_{4,1}$.
\end{proposition}
\begin{proof}
From the BKLC library of MAGMA, we obtain a quaternary $[20, 10, 8]$ code. Using Theorem~\ref{construct-LCD}, we construct a quaternary $[21, 10, 8]$ LCD code with generator matrix
    \[
    \setlength{\arraycolsep}{1.2pt}
    \left[
    \begin{array}{cccccccccccccccccccc|c}
    1 & 1 & 1 & 1 & 1 & 1 & 1 & 1 & 1 & 1 & 1 & 1 & 1 & 1 & 1 & 1 & 1 & 1 & 1 & 1 ~& ~1 \\ \hline
    0 & 1 & 0 & 0 & 0 & 0 & 0 & 0 & 0 & 0 & \zeta & 0 & \zeta^2 & 1 & \zeta & 0 & 1 & \zeta & \zeta^2 & \zeta^2 ~&~ 1 \\
    0 & 0 & 1 & 0 & 0 & 0 & 0 & 0 & 0 & 0 & \zeta^2 & 0 & 0 & 1 & \zeta^2 & \zeta^2 & 1 & 1 & \zeta^2 & 0 ~& ~\zeta^2 \\
    0 & 0 & 0 & 1 & 0 & 0 & 0 & 0 & 0 & 0 & \zeta^2 & 1 & 0 & \zeta & \zeta^2 & \zeta & \zeta & 1 & 0 & \zeta^2 ~&~ 1 \\
    0 & 0 & 0 & 0 & 1 & 0 & 0 & 0 & 0 & 0 & 0 & \zeta^2 & 1 & 0 & \zeta & \zeta^2 & \zeta & \zeta & 1 & \zeta^2 ~&~ \zeta^2 \\
    0 & 0 & 0 & 0 & 0 & 1 & 0 & 0 & 0 & 0 & 1 & \zeta^2 & \zeta^2 & \zeta & \zeta^2 & 0 & 1 & \zeta & 0 & \zeta ~&~ 0 \\
    0 & 0 & 0 & 0 & 0 & 0 & 1 & 0 & 0 & 0 & 0 & 1 & \zeta^2 & \zeta^2 & \zeta & \zeta^2 & 0 & 1 & \zeta & \zeta ~&~ 1 \\
    0 & 0 & 0 & 0 & 0 & 0 & 0 & 1 & 0 & 0 & \zeta & 1 & 1 & \zeta & \zeta & 1 & 0 & 0 & \zeta & 0 ~&~ \zeta \\
    0 & 0 & 0 & 0 & 0 & 0 & 0 & 0 & 1 & 0 & \zeta & \zeta^2 & 1 & 0 & \zeta^2 & 1 & \zeta & 0 & \zeta^2 & \zeta ~&~ 1 \\
    0 & 0 & 0 & 0 & 0 & 0 & 0 & 0 & 0 & 1 & \zeta^2 & 0 & \zeta^2 & \zeta^2 & \zeta & \zeta & 0 & \zeta & 1 & 1 ~&~ 0
    \end{array}
    \right],
    \]
where $\zeta$ is a primitive element of $\mathbb{F}_4$. The optimality follows from~\cite{Lu-quaternary-LCD}, and this code is inequivalent to the code in the BKLC database.
\end{proof}

We summarize our new codes in Table~\ref{tab:LCD_table-2}. Here, the first column gives the code name, the second column denotes the field size, the third column denotes our new parameters, the fourth column gives code parameters from BKLC (Magma)  before the embedding, the fifth column denotes the dimension of the hull of BKLC, and the last column gives the previously known parameters with references. For the new parameter sets \([23,4,14]\) and \([23,5,12]\) listed in Table~\ref{tab:LCD_table-2}, we found 226 and 38 inequivalent codes, respectively. The corresponding data can be found at~\cite{An-git}.

\begin{table}[ht]

\centering

\small

\begin{tabular}{c|c|c|c|c||c}

Code & $q$ & New parameters & BKLC & $\ell$  & Known parameters\\

\hline

$\mathcal{C}_{3, 1}$ & 3 & $[23, 4, {\bf{14}}]$ & $[19, 4, 12]$ & 4 & $[23, 4, 13]$~\cite{Li-optimal-LCD-2024-1}\\

$\mathcal{C}_{3, 2}$ & 3 & $[23, 5, {\bf{12}}]$ & $[19, 5, 11]$ & 4 & $[23, 5, 11]$~\cite{Li-optimal-LCD-2024-1}\\

$\mathcal{C}_{3, 3}$ & 3 & $[24, 6, {\bf{12}}]$ & $[20, 6, 10]$ & 4 & $[24, 6, 11]$~\cite{Li-optimal-LCD-2024-1}\\

$\mathcal{C}_{3, 4}$ & 3 & $[25, 5, {\bf{14}}]$ & $[20, 5, 12]$ & 5 & $[25, 5, 13]$~\cite{Li-optimal-LCD-2024-1}\\

$\mathcal{C}_{4, 1}$ & 4 & $[21, 10, {\bf{8}}]$ & $[20, 10, 8]$ & 1 & $[21, 10, 7]$~\cite{Lu-quaternary-LCD}\\

\end{tabular}

\caption{New ternary and quaternary LCD codes}
\label{tab:LCD_table-2}
\end{table}

\section{Conclusion}

In this paper, we have determined the minimum number of columns required to embed a binary, ternary or quaternary linear code into an LCD code. For an $[n, k]$ code $\mathcal{C}$ whose hull has dimension $\ell$, we have proved that the length of a shortest LCD embedding of $\mathcal{C}$ is $n+\ell$. Furthermore, when a generator matrix is decomposed into its hull part and LCD part, we have shown that appending an invertible matrix to the rows corresponding to the hull and an arbitrary matrix to the remaining rows always yields a shortest LCD embedding. Using this construction, we have obtained several new ternary and quaternary LCD codes. 

However, since a shortest LCD embedding can always be obtained by using any invertible $\ell \times \ell$ matrix and any $(k-\ell)\times \ell$ matrix, the search space is excessively large, and hence there is a limitation in determining a shortest LCD embedding of a given code with the largest minimum distance. Therefore, a possible direction for future research is to develop a method for choosing the appended matrices so that the resulting code has better minimum distance, or to determine whether the appended matrices can be further restricted up to equivalence.

\section*{Acknowledgments}
Jon-Lark Kim was supported in part by the BK21 FOUR (Fostering Outstanding Universities for Research) funded by the Ministry of Education (MOE, Korea), National Research Foundation of Korea (NRF) under Grant No. 4120240415042, Basic Science Research Program through the National Research Foundation of Korea (NRF) funded by the Ministry of Science and ICT under Grant No. RS-2025-24534992, National Research Foundation of Korea under Grant No. RS-2024- NR121331, and Global - Learning \& Academic research institution for Master’s·PhD students, and Postdocs(LAMP) Program of the National Research Foundation of Korea(NRF) grant funded by the Ministry of Education (No. RS-2024-00441954).

\end{document}